\documentclass[lettersize, journal]{IEEEtran}

\IEEEoverridecommandlockouts

\usepackage{amsthm}
\usepackage{mathtools}
\usepackage{graphics}
\usepackage{graphicx}
\usepackage{amsmath,amssymb,amsfonts}
\usepackage{epsfig} 
\usepackage{cite}
\usepackage{textcomp}
\usepackage{enumitem}
\usepackage{subfigure}
\usepackage{url}
\usepackage{caption}
\usepackage{setspace}

\usepackage{color}
\definecolor{LB}{RGB}{65,105,225}
\definecolor{LP}{cmyk}{0, 0.7808, 0.4429, 0.1412}
\definecolor{LO}{RGB}{210,127,0}

\newtheorem{lemma}{Lemma}
\newtheorem{corollary}{Corollary}
\newtheorem{theorem}{Theorem}
\newtheorem{problem}{Problem}

\newtheorem{remark}{Remark}

\newtheorem{assumption}{Assumption}

\begin{document}

\title{Range-Only Dynamic Output Feedback Controller for Safe and Secure Target Circumnavigation}

\author{Anand Singh and Anoop Jain,~\IEEEmembership{Senior Member,~IEEE}
\thanks{The authors are with the Department of Electrical Engineering, Indian Institute of Technology Jodhpur, 342030, India (e-mail: {m23eec001@iitj.ac.in; anoopj@iitj.ac.in}).}}

\maketitle


\begin{abstract}
The safety and security of robotic systems are paramount when navigating hostile targets. This paper addresses the problem of circumnavigating an unknown target by a unicycle robot while ensuring it maintains a desired safe distance and remains within the sensing region around the target throughout its motion. The proposed controller design is based on a composite Lyapunov framework consisting of: (i) a quadratic Lyapunov function characterizing the target-circumnavigation objective, and (ii) a barrier Lyapunov function-based potential term enforcing the prescribed safety and sensing constraints. A notable feature of the proposed approach is that it relies exclusively on \emph{local} range measurements between the robot and the target, achieved through a dynamic output-feedback controller that treats the range as the only measurable output. Using Lyapunov-based stability analysis together with an equivalent autonomous reformulation of the closed-loop dynamics, we establish asymptotic stability of the desired equilibrium while guaranteeing satisfaction of the prescribed safety and sensing constraints under the proposed controllers. Restrictive bounds on the post-design signals are also derived, and numerical simulations are presented to validate the proposed approach.
\end{abstract}



\section{Introduction}
With the increasing sophistication of modern targets, the classical problem of target circumnavigation has attracted considerable attention in the control and robotics communities \cite{cao2021safe,target3D}. A critical requirement in circumnavigating hostile targets using autonomous robotic systems, such as unmanned aerial vehicles (UAVs) and unmanned ground vehicles (UGVs), is to constrain their trajectories within a prescribed workspace such that the robot maintains a safe distance from the target while simultaneously avoiding excessive separation that may compromise sensing or communication capabilities \cite{target3D}. For example, while tracking an explosive target, a UAV must remain sufficiently far to avoid triggering the explosion, yet sufficiently close to ensure reliable monitoring \cite{zengin2007real}. Similarly, in environments involving hazardous chemical spills, radioactive regions, or critical search zones, a UGV must avoid dangerous proximity to prevent contamination while remaining within the sensing range of the target \cite{brinon2015distributed,demetriou2020navigating}. Motivated by such applications, this paper addresses the problem of safe and secure circumnavigation of an unknown stationary target by a unicycle robot while ensuring that its trajectories remain confined within a prescribed annular region surrounding the target. Throughout the paper, the term \emph{safe} refers to maintaining a desired standoff distance from the target, whereas \emph{secure} refers to ensuring that the robot remains within the prescribed sensing region around the target at all times.  

During risky target-circumnavigation missions, robotic systems often operate in harsh environments where global positioning system (GPS) signals are unreliable or entirely unavailable \cite{cao2015uav}. In such scenarios, the robot must rely solely on local information to achieve target circumnavigation while satisfying the prescribed safety and sensing constraints. Existing approaches in this direction can broadly be categorized into: (i) bearing-based methods \cite{cao2021safe,li2019safe,deghat2014localization,zheng2015enclosing,dou2020target}, and (ii) range-only-based methods \cite{cao2015uav,wang2024target,dong2019circumnavigation,matveev2016range,matveev2017range,shames2011circumnavigation}. Bearing-based approaches utilize passive sensors, such as cameras or directional antennas, to estimate the line-of-sight angle to the target. However, such methods do not directly provide distance information \cite{liu2024distance} and are often susceptible to occlusion and measurement noise in cluttered environments \cite{shames2012analysis}. In contrast, range-only approaches employ active sensors, such as LiDAR, sonar, or ultrasonic devices, to directly measure the relative distance to the target, and are generally more reliable and accurate in GPS-denied environments \cite{djugash2006range,wang2021mobile}. Motivated by these practical considerations, this work focuses on the target-circumnavigation problem using only range measurements.

In the literature, the term ``range-based" is also commonly used to describe approaches that rely on both range and range-rate measurements \cite{hashemi2015unmanned,milutinovic2017coordinate,dong2020target}. From an engineering standpoint, range-rate information can, in principle, be obtained from range measurements using numerical differentiation techniques. However, such approximations are typically sensitive to measurement noise, may induce undesirable oscillations, and can even destabilize the closed-loop system \cite{wang2021mobile,dong2021coordinate}. Although \cite{matveev2017range} proposed a derivative-free controller, the heading angle was employed directly as the control input instead of the standard angular velocity input. Other works \cite{cao2015uav,dong2021coordinate} utilized second-order sliding mode observers to estimate the range-rate signal, resulting in discontinuous controllers with relatively high computational complexity. Some studies \cite{cao2013circumnavigation,jain2022encirclement} indirectly reconstructed range-rate information using geometric relationships between range and bearing measurements. However, as discussed earlier, such approaches inherit the limitations associated with bearing-based sensing and may suffer from reduced estimation accuracy. Unlike these methods, the proposed approach employs a dynamic output-feedback framework to design a purely range-only-based controller for the target-circumnavigation problem addressed in this paper. 

It is noteworthy that none of the aforementioned works explicitly address the safety and security aspects of robotic systems deployed for target circumnavigation, which constitutes the primary focus of the proposed work. Although safe target circumnavigation has recently attracted attention, existing controller designs either rely on bearing measurements and are limited to single-integrator robot models \cite{cao2021safe,li2019safe}, or require both range and range-rate measurements \cite{bhati2025safe}. A recent work \cite{target3D} proposed a robust controller for a higher-order vehicle model to enclose a maneuvering target in a 3D engagement scenario under full-state sensing. In contrast, this paper considers unicycle robots operating in a 2D environment under \emph{range-only} sensing by leveraging logarithmic asymmetric barrier Lyapunov functions (ABLFs) \cite{tee2009barrier}. The main contributions and distinguishing features of this work are summarized as follows:

\begin{itemize}[leftmargin=*]
\item \emph{Summary of Contributions:} By combining a quadratic Lyapunov function characterizing the target-circumnavigation objective with an ABLF-based Lyapunov function enforcing the prescribed safety and security constraints, we first develop a stabilizing feedback controller utilizing both range and range-rate information. Motivated by this design, we subsequently propose a controller that relies exclusively on range measurements. For the latter case, a dynamic output-feedback observer is constructed using the available range measurement and an auxiliary state variable, and an additional Lyapunov term is incorporated into the composite Lyapunov framework to account for range-rate estimation. Since the resulting closed-loop dynamics are nonautonomous due to the time-varying linear speed of the robot, a suitable time-reparametrization is introduced to obtain an equivalent autonomous representation, thereby enabling a rigorous stability analysis using LaSalle's invariance principle. Using this framework, we establish asymptotic stability of the closed-loop system and prove that the robot trajectories satisfy the prescribed safety and security constraints under both proposed controllers. Furthermore, restrictive bounds on the post-design signals are derived in terms of the composite Lyapunov functions, and boundedness of the controllers is established. 	
\item \emph{Main Features:} The key distinguishing features of the proposed approach are summarized as follows: (a) Unlike \cite{wang2024target}, the turn-rate is treated as the \emph{only} control input, while the unicycle robot is assumed to move with a \emph{nonzero} linear speed that is not subject to control. This assumption is motivated by practical considerations related to endurance and control efficiency in robotic platforms \cite{target3D,jain2019trajectory,cao2015uav}, (b) Unlike \cite{jain2019trajectory}, no restriction is imposed on the robot's initial heading angle for guaranteeing safety and security, (c) In contrast to \cite{wang2024target}, the proposed framework admits a broad class of potentially unbounded \emph{design functions}. Despite this flexibility, the proposed controllers remain finite (although their bounds depend on the initial conditions) due to the safety and security constraints inherently embedded into the controller design, (d) Distinct from the backstepping-based methodology in \cite{tee2009barrier}, the proposed approach employs a dynamic output-feedback observer-controller architecture relying solely on range measurements and avoids the use of higher-order derivatives. Moreover, the resulting stability analysis is relatively simple and is carried out using LaSalle’s invariance principle through an equivalent autonomous reformulation of the closed-loop dynamics, rather than Barbalat's lemma-based arguments as in \cite{tee2009barrier,wang2024target}.
\end{itemize}

\paragraph*{Preliminaries}
We denote the set of real numbers by $\mathbb{R}$, the set of non-negative real numbers by $\mathbb{R}_+$, and the $n$-dimensional real vector space by $\mathbb{R}^n$. The Euclidean norm of $q \in \mathbb{R}^n$ is represented by $\|q\|$, and the superscript $q^\top$ denotes its transpose. The vector $\pmb{0}_n \in \mathbb{R}^n$ represents an $n$-dimensional vector with all components equal to zero. The symbol $\mathcal{C}^n$ is used to denote continuously differentiable functions of order $n$. For a differential map $f: \mathbb{D} \to \mathbb{R}$, $\mathbb{D} \subset \mathbb{R}^n$, $\nabla_q f= [\partial f/\partial q_1, \cdots, \partial f/\partial q_n]^\top$ is the gradient with respect to $q = [q_1, \cdots, q_n]^\top \in \mathbb{R}^n$. A continuous, strictly increasing function $\alpha: \mathbb{R}_+ \to \mathbb{R}_+$ is said to belong to class $\mathcal{K}_{\infty}$ if $\alpha(0) = 0$ and $\alpha(\bullet) \to \infty$ as $\bullet \to \infty$ \cite[pg. 144]{khalil2002control}. Let $\Omega: \mathbb{R} \to \mathbb{R}$ be a class of smooth functions (now onward referred to as \emph{design functions}) satisfying the properties: {\bf (P1)} $\Omega(0) = 0$ and $\mu\Omega(\mu) > 0 $ for $\mu \neq 0$; {\bf (P2)} $\alpha_1(|\mu|) \leq \int_{0}^{\mu} \Omega(s) ds \leq \alpha_2(|\mu|)$ for some class $\mathcal{K}_{\infty}$ functions $\alpha_1$ and $\alpha_2$. Note that it follows from {\bf (P1)} that $\int_{0}^{\mu} \Omega(s) ds > 0$ for all $\mu \neq 0$ and $\int_{0}^{\mu} \Omega(s) ds = 0$ only if $\mu = 0$. Examples of such functions include $\mu^n$ where $n > 0$ is an odd integer, $\mu/\sqrt{1 + \mu^2}$, $\tanh(\mu)$, $\arctan(\mu)$, etc. For clarity, the time argument is often omitted when evident from the context. Next, we state the following lemma dealing with ABLF, which will be helpful in the subsequent analysis. 

\begin{lemma}[Convergence under ABLF \cite{tee2009barrier}]\label{blf_lemma}
For any positive constants $k_a $ and $ k_b$, let $\mathcal{Z} \triangleq \{\xi \in \mathbb{R}: -k_a < \xi < k_b\} \subset \mathbb{R}$ and $\mathcal{N} \triangleq \mathbb{R}^\ell \times \mathcal{Z} \subset \mathbb{R}^{\ell+1}$ be the open sets. Consider the system $\dot{\eta} = h(t, \eta) $, where $\eta \triangleq [\xi, w]^\top \in \mathcal{N}$, and $ h: \mathbb{R}_+ \times \mathcal{N} \to \mathbb{R}^{\ell+1}$ is piecewise continuous in $t$ and locally Lipschitz in $\eta$, uniformly in $t$, on $\mathbb{R}_+ \times \mathcal{N}$. Suppose there exist functions $U : \mathcal{Z} \to \mathbb{R}_+$ and $W : \mathbb{R}^\ell \to \mathbb{R}_+ $, continuously differentiable and positive definite in their respective domains, such that $U(\xi) \to \infty $ as $\xi \to -k_a$ or $\xi \to k_b$ and $\gamma_1(\|w\|) \leq W(w) \leq \gamma_2(\|w\|)$ for some class $\mathcal{K}_\infty$ functions $\gamma_1$ and $\gamma_2$. Let $V(\eta) \triangleq U(\xi) + W(w)$ and assume that $\xi(0) \in \mathcal{Z}$. If $\dot{V} = (\nabla V)^T h \leq 0 $ holds in the set $\xi \in \mathcal{Z}$, then $\xi(t) \in \mathcal{Z}$ for all $ t \in [0, \infty)$.
\end{lemma}

\begin{remark}\label{remark_design_function}
An earlier work \cite{wang2024target} also utilized the design functions $\Omega$ in the control scheme. However, $\Omega$ being a saturation function in \cite{wang2024target} need not satisfy {\bf (P2)}; instead, in addition to {\bf (P1)}, it satisfies two properties: {\bf (P3)} $\int_{0}^{\mu} \Omega(s) ds \to \infty, \ \text{as} \ \mu \to \infty$, and {\bf (P4)} $|\Omega(\mu)| \leq M_1 \ \text{and} \ |d\Omega(\mu)/d\mu| \leq M_2$, for some positive constants $M_1$ and $M_2$. However, we do not impose {\bf (P4)} in this paper, and allow $\Omega$ to be unbounded, provided {\bf (P1)} and {\bf (P2)} hold. Note that {\bf (P3)} is inherited in {\bf (P2)}, since $\alpha_1$ and $\alpha_2$ belong to the class $\mathcal{K}_\infty$ functions. For instance, both $\Omega_1(\mu) = \mu$ (an unbounded function) and $\Omega_2(\mu) = {\mu}/{\sqrt{1 + \mu^2}}$ (a saturation function with $|\Omega_2(\mu)| < 1$) satisfy {\bf (P1)} and {\bf (P2)}. However, only $\Omega_2(\mu)$ is admissible in \cite{wang2024target}, as $\Omega_1(\mu)$ does not satisfy {\bf (P4)}.
\end{remark}

 \begin{figure}[t]
	\centering{
		\includegraphics[width=6.0cm]{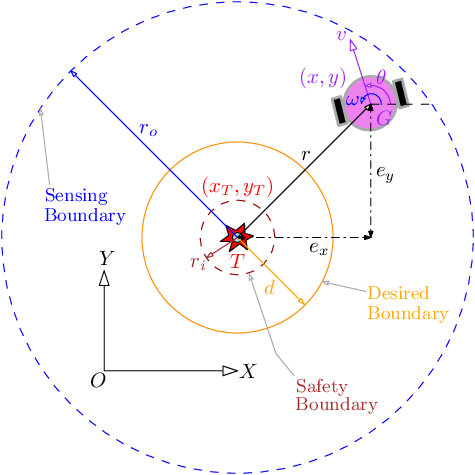}
		\caption{Safe and secure target circumnavigation by a unicycle robot (\emph{castor} wheel is not shown for simplicity).}
		\label{fig_problem_description}}
\end{figure}


\section{System Model, Problem Description and Solution Methodology}

\subsection{System Model and Problem Description}
Consider a unicycle-type robot moving in a planar space $\mathbb{R}^2$. The robot's kinematics are governed by:
\begin{equation}\label{system_model}
\dot{x}(t) = v(t)\cos\theta(t), \ \dot{y}(t) = v(t)\sin \theta(t), \ \dot{\theta}(t) = \omega(t),
\end{equation}
where $[x(t), y(t)]^\top \in \mathbb{R}^2$ and $\theta(t) \in [0, 2\pi)$ denote the position and heading angle in the global $(X, Y)-$coordinate frame, respectively. The term $v(t)$ represents a \emph{given} linear velocity profile and $\omega(t) \in \mathbb{R}$ is the turn-rate control input. The objective is to design $\omega(t)$ such that \eqref{system_model} circumnavigates a stationary target $T$ located at position $[x_T, y_T]^\top \in \mathbb{R}^2$ (as shown in Fig.~\ref{fig_problem_description}). We assume that the robot does not know the position of the target $T$, instead it can measure the line-of-sight (LoS) distance (or range) to it, denoted as: $r(t) = \sqrt{(x(t) - x_T)^2 + (y(t) - y_T)^2}$. Note that the global coordinate frame shown in Fig.~\ref{fig_problem_description} is introduced solely for theoretical analysis. It plays no role in the actual control implementation since the controller $\omega(t)$ relies exclusively on the range measurement $r$, as discussed subsequently.

As illustrated in Fig.\ref{fig_problem_description}, our objective in this paper is to allow the robot to circumnavigate the target $T$ at a constant radial distance $d$, while assuring that it $(a)$ maintains a safe distance $r_{\text{safe}} \triangleq r_i > 0$ from the target, and $(b)$ remains within a desired sensing region defined by the outer radial distance $r_{\text{sensing}} \triangleq r_o > 0$. Collectively, it is desired that $r(t) \to d$ as $t \to \infty$, and $r_i < r(t) < r_o$ for all $t \geq 0$ where $r_o > r_i$. According to Fig.~\ref{fig_problem_description}, it is clear that the three radii satisfy the condition $r_i < d < r_o$. It is imperative that the robot must start within the region enclosed by the two radii, that is, $r_i < r(0) < r_o$, which is formally stated in the following assumption:

\begin{assumption}\label{assumption}
For the given inner and outer radii $r_i$ and $r_o$, respectively, let $\mathcal{Z}_r \triangleq \{r \in \mathbb{R}_+ \mid r_i < r < r_o\}$ be an open set. The robot \eqref{system_model} begins its motion such that $r(0) \in \mathcal{Z}_r$. 
\end{assumption}

It is worth noting that Assumption~\ref{assumption} places no restrictions on the initial heading angle $\theta(0)$ of \eqref{system_model}, unlike \cite{jain2019trajectory}. For practical reasons, we assume that the robot moves with the \emph{forward} linear speed $v(t)$ satisfying the following assumption:

\begin{assumption}\label{assumption_velocity}
The linear velocity profile $v(t)$ in \eqref{system_model} is continuously differentiable and satisfies $0 < v_{\min} \leq v(t) \leq v_{\max} < \infty$ for all $t \ge 0$, where $v_{\min}$ and $v_{\max}$ are known positive constants.
\end{assumption}

We now formally state the problem addressed in this paper: 

\begin{problem}\label{problem}
Consider the robot-target engagement as shown in Fig.~\ref{fig_problem_description} and suppose that Assumption~\ref{assumption} and Assumption~\ref{assumption_velocity} hold. Design the controller $\omega(t)$ in \eqref{system_model} such that $r(t) \to d$ as $t \to \infty$ and $r_i < r(t) < r_o$ for all $t \geq 0$.
\end{problem}

\subsection{Solution Methodology}
Let $e_x = x - x_T$ and $e_y = y - y_T$ be the positional errors between the robot's geometric center $G$ and target $T$ (see Fig.~\ref{fig_problem_description}). Using \eqref{system_model}, the error dynamics are obtained as $\dot{e}_x = \dot{x} = v\cos\theta$ and $\dot{e}_y = \dot{y} = v\sin\theta$, since $\dot{x}_T = \dot{y}_T = 0$. Notice that the controller $\omega$ does not appear in either of the error dynamics, and hence, the errors $[e_x, e_y]^\top$ cannot be \emph{directly} influenced by $\omega$. To address this, we solve Problem~\ref{problem} by proposing the following coordinate transformation, which maps the error vector $[e_x, e_y]^\top$ into a new coordinate frame $[\bar{e}_x, \bar{e}_y]^\top$, as shown in Fig.~\ref{fig_problem_description}. Specifically, we define a virtual point $P$ located at a fixed distance $d$ from the robot, positioned perpendicularly to the robot's heading direction, either to the left or right of the robot. The errors $[\bar{e}_x, \bar{e}_y]^\top$ are defined such that $PP'$ is perpendicular to both $TP'$ and $PG$. We design the controller by observing the motion of $P$ relative to the target, rather than directly controlling the robot's centroid \cite{wang2024target}. In Fig.~\ref{fig_problem_description}, note that $P$ may even fall outside the outer boundary if considered in the right. However, this does not create any problem for designing the controller, as it is the virtual point. According to geometry in Fig.~\ref{fig_problem_description}, one can write
\begin{equation} \label{transformation}
	\begin{bmatrix}
		\bar{e}_x\\
		\bar{e}_y
	\end{bmatrix}
	= \begin{bmatrix}
		\cos\theta & \sin\theta\\
		-\sin\theta & \cos\theta
	\end{bmatrix}
	\begin{bmatrix}
		{e}_x\\
		{e}_y
	\end{bmatrix}
	+ d \begin{bmatrix}
		0 \\
		1
	\end{bmatrix}.
\end{equation}

From \eqref{transformation}, the error dynamics in the transformed coordinates can be expressed as $\dot{\bar{e}}_x = \dot{e}_x\cos\theta + \dot{e}_y\sin\theta + (e_y\cos\theta - e_x\sin\theta)\dot{\theta}$ and $\dot{\bar{e}}_y = -\dot{e}_x\sin\theta + \dot{e}_y\cos\theta - (e_x\cos\theta + e_y\sin\theta)\dot{\theta}$. Using \eqref{system_model} and \eqref{transformation}, these reduce to \cite{wang2024target}:
\begin{equation}\label{error_dynamics_transformed}
\dot{\bar{e}}_x = v + (\bar{e}_y - d) \omega, \quad \dot{\bar{e}}_y = -\bar{e}_x \omega.
\end{equation}
Following \eqref{error_dynamics_transformed}, our solution approach focuses on designing the controller $\omega$ in the new coordinates to solve Problem~\ref{problem}. Substituting $[e_x, e_y]^\top$ in terms of $[\bar{e}_x, \bar{e}_y]^\top$ from \eqref{transformation}, one can observe that 
\begin{equation}\label{r_square}
	r^2 = e_x^2 + e_y^2 = \bar{e}_x^2 + (\bar{e}_y - d)^2,
\end{equation}
and the time-derivative of \eqref{r_square}, using \eqref{transformation} and \eqref{error_dynamics_transformed}, is given by:
\begin{equation}\label{range_dynamics}
	(1/2)({d} r^2/{dt}) = r\dot{r} = (\bar{e}_x\dot{\bar{e}}_x + (\bar{e}_y - d)\dot{\bar{e}}_y) = v\bar{e}_x,
\end{equation}
where $\dot{r}$ denotes the range-rate. From \eqref{r_square}, it readily follows that if $[\bar{e}_x, \bar{e}_y]^\top \to \pmb{0}_2$, $r \to d$. Note that if the direction of the unit vector $[0, 1]^\top \in \mathbb{R}^2$ is reversed in the transformation \eqref{transformation}, one gets a similar formulation as in \eqref{error_dynamics_transformed} and \eqref{r_square} with a change in the sign of $d$, which is excluded in our discussion for simplicity. We further introduce the radial error 
\begin{equation}\label{radial_error}
e_r(t) \triangleq r(t) - d,	
\end{equation}
representing the deviation between the robot's position and the radius of the desired circumnavigation path in the polar coordinate. By appropriately constraining \eqref{radial_error}, we aim to satisfy the desired safety and security constraints on the robot's motion. Moreover, $e_r \to 0 \implies r \to d$, as desired. Therefore, the equivalent control objectives to Problem~\ref{problem} is to design controller $\omega$ in \eqref{system_model} such that the transformed errors $[\bar{e}_x, \bar{e}_y]^\top \to \pmb{0}_2$ as $t \to \infty$ and $r_i - d < e_r(t) < r_o - d$ (alternatively, $r_i < r(t) < r_o$) for all $t \geq 0$. 


\section{Lyapunov Functions and Controller Design}\label{lyapunov}
The proposed controller design methodology is based on a composite Lyapunov framework consisting of: (i) a quadratic Lyapunov function for driving the transformed errors $\bar{e}_x$ and $\bar{e}_y$ to the origin, and (ii) an ABLF associated with the radial error \eqref{radial_error}, which guarantees confinement of $r(t)$ within the prescribed safety and sensing bounds. Below, we describe these functions separately, followed by the controller design.  

\subsection{Constructions of Lyapunov Functions}   
To facilitate the design of controller $\omega$ that derive errors $\bar{e}_x$ and $\bar{e}_y$ to the origin, we introduce the following quadratic Lyapunov function: 
\begin{equation}\label{Ve}
V_e(\bar{e}_x, \bar{e}_y) = ({1}/{2})(\bar{e}_x^2 + \bar{e}_y^2),
\end{equation}
which is positive-definite and attains its minimum value of zero only if $[\bar{e}_x, \bar{e}_y]^\top \to \pmb{0}_2$. The time derivative of \eqref{Ve}, along the error dynamics \eqref{error_dynamics_transformed}, is given by
\begin{equation}\label{Ve_derivative}
\dot{V}_e = \bar{e}_x\dot{\bar{e}}_x + \bar{e}_y\dot{\bar{e}}_y = (v - d \omega)\bar{e}_x.
\end{equation}
On the other hand, to constrain the range $r(t)$ within the predefined barriers, we propose the following ABLF \cite{tee2009barrier}: 
\begin{equation}\label{Vr}
\hspace*{-1pt}V_r(e_r) = \frac{\sigma(e_r)}{2} \ln \left[\frac{\delta_b^2}{\delta_b^2 - e_r^2} \right] + \frac{1 - \sigma(e_r)}{2} \ln \left[\frac{\delta_a^2}{\delta_a^2 - e_r^2}\right],
\end{equation}
where ``$\ln"$ denotes the natural logarithm function and $\sigma(e_r)$ is a logic parameter, defined as: $\sigma(e_r) \triangleq 1$ if $e_r > 0$, and $\sigma(e_r) \triangleq 0$ if $e_r \leq 0$. Further, $\delta_a$ and $\delta_b$ are positive constants and rely on the radii $r_i$, $r_o$ and $d$ according to the relations:
\begin{equation}\label{barriers}
\delta_a = d - r_i, \qquad \delta_b = r_o - d. 
\end{equation}
Note that the potential \eqref{Vr} is positive-definite and $\mathcal{C}^1$ in the domain $e_r \in (-\delta_a, \delta_b)$ (please refer to \cite[Lemma~2]{tee2009barrier} for a proof) and $V_r(e_r) = 0$ only if $e_r = 0$. Thus, $V_r(e_r)$ is a valid Lyapunov function candidate in the domain $e_r \in (-\delta_a, \delta_b)$. The time derivative of \eqref{Vr}, along the dynamics of \eqref{radial_error}, is obtained as (note that $\sigma$ is a logic parameter): $\dot{V}_r = e_r \dot{e}_r [\frac{\sigma(e_r)} {\delta_b^2 - e_r^2} + \frac{1 - \sigma(e_r)}{\delta_a^2 - e_r^2}] = (r-d)\dot{r} \eta(e_r)$, where we denote
\begin{equation}\label{eta}
\eta(e_r) \triangleq \frac{\sigma(e_r)}{\delta_b^2 - e_r^2} + \frac{1 - \sigma(e_r)}{\delta_a^2 - e_r^2}.
\end{equation}
Next, substituting $\dot{r}$ from \eqref{range_dynamics}, we get
\begin{equation}\label{Vr_derivative}
	\dot{V}_r = [(r - d)/r] v\bar{e}_x \eta(e_r). 
\end{equation}
Note that $e_r$ is essentially a function of range $r(t)$ (see \eqref{radial_error}), and hence, $\eta(e_r)$ in \eqref{eta}. Following this, we will use $\eta(e_r) \triangleq \eta(r)$ interchangeably in the following analysis for simplicity.    

\subsection{Controller Design}
In this subsection, we propose the controller $\omega$ based on the potentials \eqref{Ve} and \eqref{Vr}, and solve Problem~\ref{problem}. We first present a motivating controller that utilizes both range and range-rate information. Subsequently, we extend this design to develop a controller that relies solely on range measurements using a dynamic output-feedback framework.

As it will be shown later, the proposed controller admits the form $\omega = v(t)u$, where $u \in \mathbb{R}$ is purely state-dependent. Consequently, the closed-loop system \eqref{system_model} can be represented as $[\dot{x}(t), \dot{y}(t),  \dot{\theta}(t)]^\top = v(t)[\cos\theta(t), \sin\theta(t), u(t)]^\top$, where $v(t)$ appears as a \emph{common} multiplicative factor. Motivated by this structure, we next state a lemma, which establishes the \emph{one-to-one} correspondence between the trajectories of two systems of the form $d\zeta(t)/dt = v(t)f(\zeta(t))$ and $d\zeta(\tau)/d\tau = f(\zeta(\tau))$ under positive time-reparametrization. This property is commonly referred to as \emph{trajectory/orbital-equivalence} \cite{sampei2003time,aguado2019extended}.

\begin{lemma}\label{lem_trajectory_equivalence}
Consider the non-autonomous system $d\zeta/dt = v(t)f(\zeta)$, where $\zeta \in \mathbb{R}^n$, $f: \mathbb{R}^n \to \mathbb{R}^n$ is locally Lipschitz, and Assumption~\ref{assumption_velocity} holds. Define the time-reparametrization
\begin{equation}\label{time_reparametrization}
\tau(t) = \int_{0}^{t} v(s) ds.
\end{equation}
Then the original system is trajectory-equivalent to the autonomous system $d\zeta/d\tau = f(\zeta)$. Moreover, if $\zeta(\tau) \to \zeta^{\star}$ as $\tau \to \infty$, then $\zeta(t) \to \zeta^{\star}$ as $t \to \infty$. 
\end{lemma}

\begin{proof}
Please refer to Appendix.
\end{proof}

Note that Lemma~\ref{lem_trajectory_equivalence} implies that a positive and monotonically increasing time-reparametrization alters only the rate at which the state evolves along a trajectory, while leaving the corresponding orbit in the phase space unchanged. Using \eqref{time_reparametrization}, \eqref{system_model} can be expressed in reparametrized time scale $\tau$ as: 
\begin{equation}\label{system_model_reparametrized}
\dot{x}(\tau) = \cos\theta(\tau), \quad \dot{y}(\tau) = \sin \theta(\tau), \quad \dot{\theta}(\tau) = u(\tau),
\end{equation}
which corresponds to a unit-speed unicycle with input $u$ and represents an autonomous system. For system \eqref{system_model_reparametrized}, the corresponding error dynamics \eqref{error_dynamics_transformed} and \eqref{range_dynamics} can be analogously obtained in the $\tau$-scale as: 
\begin{equation}\label{error_dynamics_xy_reparametrized}
	\dot{\bar{e}}_x(\tau) = 1 + (\bar{e}_y(\tau) - d) u, \quad \dot{\bar{e}}_y(\tau) = -\bar{e}_x(\tau) u, 
\end{equation}
\begin{equation} \label{error_dynamics_r_reparametrized}
	r(\tau)\dot{r}(\tau) = \bar{e}_x(\tau),
\end{equation}
respectively, while \eqref{r_square} remains unchanged in $\tau$-scale. The above autonomous reformulation, together with Lemma~\ref{lem_trajectory_equivalence}, forms the basis for the subsequent stability analysis. For notational simplicity, the overdot notation in the following theorems denotes differentiation with respect to $\tau$-scale.

\subsubsection{Availability of range and range-rate measurements} 
When both range and range-rate are available, we have the following result:

\begin{theorem}\label{thm_control_design_1}
	Consider the robot-target engagement shown in Fig.~\ref{fig_problem_description} concerning the robot model \eqref{system_model_reparametrized} obtained under time-reparametrization \eqref{time_reparametrization}. Suppose that Assumption~\ref{assumption} and Assumption~\ref{assumption_velocity} hold. Let the model \eqref{system_model_reparametrized} be governed by the controller
	\begin{equation}\label{control_1_reparametererized}
		u = (1/d) + k_1\Omega(\bar{e}_x) + k_2 [(r - d)/r]\eta(r),
	\end{equation}
	where $k_1 > 0$ and $k_2 > 0$ are the controller gains, and the design function $\Omega(\bullet)$ satisfies properties {\bf (P1)} and {\bf (P2)} in the preliminaries. Then, $r_i < r(\tau) < r_o$ for all $\tau \geq 0$, and $\bar{e}_x(\tau) \to 0$, $\bar{e}_y(\tau) \to 0$ and $r(\tau) \to d$ as $\tau \to \infty$, for all initial conditions in $\mathcal{Z}_r$.  
\end{theorem}

Note that the set $\mathcal{Z}_r$ in Assumption~\ref{assumption} is independent of time and remains unchanged under time-reparametrization \eqref{time_reparametrization}.

\begin{proof}
Consider the joint candidate Lyapunov function
	\begin{equation}\label{V1}
		V_1(\bar{e}_x, \bar{e}_y, e_r) = V_e(\bar{e}_x, \bar{e}_y) + dk_2V_r(e_r),
	\end{equation}
	where $V_e$ and $V_r$ are defined in \eqref{Ve} and \eqref{Vr}, respectively. Analogous to \eqref{Ve_derivative} and \eqref{Vr_derivative}, one obtains $\dot{V}_e = (1 - d u)\bar{e}_x$ and $\dot{V}_r = [(r - d)/r] \bar{e}_x \eta(e_r)$ in $\tau$-scale. Therefore, the time derivative of \eqref{V1} in $\tau$-scale is given by $\dot{V}_1 = (1 - du)\bar{e}_x + dk_2 [(r - d)/r] \bar{e}_x \eta(r)$. Further, substituting $u$ from \eqref{control_1_reparametererized} and simplifying, yields $\dot{V}_1 = -dk_1 \bar{e}_x \Omega(\bar{e}_x)$. Since $d > 0, \ k_1 > 0$ and $\Omega(\bullet)$ satisfies property {\bf (P1)}, we conclude that $\dot{V}_1$ is negative semi-definite, i.e., $\dot{V}_1 \leq 0$. As Assumption~\ref{assumption} holds, it now immediately follows from Lemma~\ref{blf_lemma} that $-\delta_a < e_r(\tau) < \delta_b$ for all $\tau \geq 0$ in $\mathcal{Z}_r$, since the potential $V_r \to \infty$ as $e_r \to -\delta_a$ or $e_r \to \delta_b$ and $c_1(\bar{e}_x^2 + \bar{e}_y^2) \leq V_e \leq c_2(\bar{e}_x^2 + \bar{e}_y^2)$ for some constants $c_1 \in [0, 1/2]$ and $c_2 \in [1/2, \infty)$. Substituting for $\delta_a$ and $\delta_b$ from \eqref{barriers}, and $e_r(\tau)$ from \eqref{radial_error}, the error inequality $-\delta_a < e_r(\tau) < \delta_b$ implies that $r_i < r(\tau) < r_o$ for all $\tau \geq 0$ in $\mathcal{Z}_r$. This proves the first statement.
	
	To prove the second statement, notice that $V_1(\bar{e}_x, \bar{e}_y, e_r) \leq V_1(\bar{e}_x(0), \bar{e}_y(0), e_r(0))$ for all $\tau \geq 0$ in $\mathcal{Z}_r$, because $\dot{V}_1 \leq 0$. Further, it is straightforward to see that $V_1(\bar{e}_x(0), \bar{e}_y(0), e_r(0))$ is finite and positive for any given initial conditions in $\mathcal{Z}_r$. In other words, it implies that there exists a positive constant $\lambda_1$ such that $V_1(\bar{e}_x, \bar{e}_y, e_r) \leq \lambda_1$ for all initial conditions in $\mathcal{Z}_r$. Note that the set $\Lambda_1 \triangleq \{[\bar{e}_x, \bar{e}_y, e_r]^\top \in \mathbb{R}^2 \times (-\delta_a, \delta_b) \mid V_1(\bar{e}_x, \bar{e}_y, e_r) \leq \lambda_1\}$ is compact and positively invariant, since $V_1$ is positive definite and continuously differentiable, and $\dot{V}_1 \leq 0$ in $\mathcal{Z}_r$. From LaSalle's invariance principle \cite[Corollary 4.1, pg. 128]{khalil2002control}, it follows that all the trajectories of \eqref{error_dynamics_xy_reparametrized} and \eqref{error_dynamics_r_reparametrized} converge to the largest invariant set $\Delta_1$ contained in the set $\Gamma_1 \triangleq \{[\bar{e}_x, \bar{e}_y, e_r] \in \mathbb{R}^2 \times (-\delta_a, \delta_b) \mid \dot{V}_1 = 0\} \subset \Lambda_1$ as $\tau \to \infty$. Notice that $\dot{V}_1 = 0$ only if $\bar{e}_x = 0$. For the system to maintain $\dot{V}_1 = 0$ condition on $\Delta_1 \subset \Gamma_1$, the trajectories of \eqref{error_dynamics_xy_reparametrized} and \eqref{error_dynamics_r_reparametrized} must lie \emph{identically} on $\Delta_1$, from which it follows that $\bar{e}_x(\tau) \equiv 0 \implies \dot{\bar{e}}_x(\tau) \equiv 0$ on $\Delta_1$. Substituting these identities into \eqref{error_dynamics_xy_reparametrized} yields 
	\begin{equation}\label{invariant_condition_1}
	1 + (\bar{e}_y(\tau) - d)u \equiv 0,	
	\end{equation}
	and $\dot{\bar{e}}_y(\tau) \equiv 0$ on $\Delta_1$. Also, it must hold from \eqref{r_square} and \eqref{error_dynamics_r_reparametrized} that $r^2(\tau) \equiv (\bar{e}_y(\tau) - d)^2$ and $\dot{r}(\tau) \equiv 0$ on $\Delta_1$, respectively, since $r(\tau) \neq 0$ (as proved in the first statement). To satisfy both these conditions \emph{simultaneously}, it is necessary that $r(\tau)$ and $\bar{e}_y(\tau)$ be \emph{identically constant} on $\Delta_1$. Moreover, because $r_i < r(\tau) < r_o$ for all $\tau \geq 0$ in $\mathcal{Z}_r$, the condition $r^2(\tau) \equiv (\bar{e}_y(\tau) - d)^2$ implies that $\bar{e}_y(\tau)$ remains bounded within $\Delta_1$.
	
	Further, notice that $\bar{e}_y - d \neq 0$ in $\Gamma_1$, because if $\bar{e}_y - d = 0$, then $r = 0$ in $\Gamma_1$ according to \eqref{r_square}, which is a contradiction since $r_i < r(\tau) < r_o$ for all $\tau \geq 0$ in $\mathcal{Z}_r$. Consequently, it follows that $r(\tau) \equiv r_c$ on $\Delta_1$ where $r_c \in (r_i, r_o)$ is a nonzero constant. This implies that the robot \eqref{system_model_reparametrized} moves along the circular trajectory with angular velocity $u = 1/r_c$ on $\Delta_1$, as $\tau \to \infty$. Substituting $u = 1/r_c$ into \eqref{control_1_reparametererized} and simplifying on $\Delta_1$, we obtain $(1/r_c) = (1/d) + k_2 [(r_c - d)/r_c] \eta(r_c)$, which reduces to the condition
	\begin{equation}\label{condition}
		(r_c - d)(1 + dk_2\eta(r_c)) \equiv 0,
	\end{equation}
	on $\Delta_1$. From \eqref{eta}, it is clear that $\eta(e_r) > 0$ for $e_r \in (-\delta_a, \delta_b)$, and hence, $1 + dk_2\eta(r_c) > 0$ on $\Delta_1$. Therefore, \eqref{condition} holds only if $r_c \equiv d$ on $\Delta_1$. Alternatively, the robot \eqref{system_model_reparametrized} moves around the circular trajectory of the radius $d$ with angular rate $u = 1/d$ on $\Delta_1$, as $\tau \to \infty$. Now, using \eqref{invariant_condition_1}, one concludes that $\bar{e}_y \equiv 0$ on $\Delta_1$, as $\tau \to \infty$. Consequently, it follows that the largest invariant set $\Delta_1 \subset \Gamma_1 \subset \Lambda_1$ contains only the desired equilibrium where $\bar{e}_x(\tau) \to 0$, $\bar{e}_y(\tau) \to 0$, and $r(\tau) \to d$, as $\tau \to \infty$. This concludes the proof.
\end{proof}

We now state the following corollary to Theorem~\ref{thm_control_design_1} in the original $t$-scale. The result follows directly from Lemma~\ref{lem_trajectory_equivalence}.

\begin{corollary}\label{cor_control_1}
Consider the robot-target engagement shown in Fig.~\ref{fig_problem_description} concerning the robot model \eqref{system_model} and suppose that Assumption~\ref{assumption} and Assumption~\ref{assumption_velocity} hold. Let the robot \eqref{system_model} be governed by the controller
\begin{equation}\label{control_1}
\omega = v(t)[(1/d) + k_1\Omega(\bar{e}_x) + k_2 [(r - d)/r]\eta(r)],
\end{equation}
where the gains are as defined in Theorem~\ref{thm_control_design_1}. Then, $r_i < r(t) < r_o$ for all $t \geq 0$, and $\bar{e}_x(t) \to 0$, $\bar{e}_y(t) \to 0$ and $r(t) \to d$ as $t \to \infty$, for all initial conditions in $\mathcal{Z}_r$.
\end{corollary}

Note that the control law \eqref{control_1} can be expressed in terms of $r$ and $\dot{r}$, using \eqref{range_dynamics}, as:
\begin{equation}\label{control_omega_final}	
\omega = v(t)\left[\frac{1}{d} + k_1\Omega\left(\frac{r\dot{r}}{v(t)}\right) + k_2 \left(\frac{r - d}{r}\right)\eta(r)\right].
\end{equation}

\subsubsection{Availability of range-only measurements}
In controller \eqref{control_omega_final}, notice that only the term $\Omega(\bullet)$ is dependent on the range-rate $\dot{r}$. Our approach here is to modify the argument of $\Omega$ via a term that is solely dependent on the range information. In this direction, we use the output dynamic feedback approach \cite{wang2024target} where we designate $\gamma = (1/2)r^2$ as the output and propose the following variation of the controller \eqref{control_omega_final}:
\begin{subequations}\label{control_2}
	\begin{align}
		\label{control_2a} \omega &= v(t)\left[\frac{1}{d} + k_1 \Omega(-\kappa z + \gamma) + k_2 \left(\frac{r - d}{r}\right)\eta(r)\right]\\
		\label{control_2b}	\dot{z} &= v(t)\left(-\kappa z + \gamma \right),
	\end{align}
\end{subequations}
\noindent where $\kappa > 0$ is a gain term. In the revised controller \eqref{control_2a}, our approach is to estimate $r\dot{r}/v$ using only the range $r$ and an auxiliary variable $z \in \mathbb{R}$ with dynamics \eqref{control_2b}. For further analysis, we denote $\beta \triangleq -\kappa z + \gamma$. Analogous to \eqref{system_model_reparametrized}, the closed-loop system corresponding to \eqref{control_2} can be expressed in the reparametrized time scale $\tau$, under \eqref{time_reparametrization}, as:
\begin{equation}\label{system_model_reparametrized_new}
\hspace*{-5pt} \dot{x}(\tau) = \cos\theta(\tau), \ \dot{y}(\tau) = \sin \theta(\tau), \ \dot{\theta}(\tau) = u(\tau), \ \dot{z}(\tau) = \beta(\tau),
\end{equation}
where $\beta$ and $u$ are state-dependent. From \eqref{control_2}, we write
\begin{equation}\label{control_2_reparametrized}
u = (1/d) + k_1 \Omega(-\kappa z + \gamma) + k_2 [(r - d)/r]\eta(r).
\end{equation}
Further, the time derivative of $\beta$ in $\tau$-scale is obtained as:
\begin{equation}\label{beta_dynamics}
	\dot{\beta} = -\kappa\dot{z} + \dot{\gamma} = -\kappa\dot{z} + r\dot{r} = -\kappa\beta + \bar{e}_x, 
\end{equation}
using \eqref{error_dynamics_r_reparametrized} and \eqref{system_model_reparametrized_new}. Also, the error dynamics \eqref{error_dynamics_xy_reparametrized}, under the controller \eqref{control_2_reparametrized}, can be written in $\tau$-scale as: 
\begin{subequations}\label{error_dynamics_new}
	\begin{align}
		\nonumber  \hspace*{-5pt} \dot{\bar{e}}_x & = 1 + (\bar{e}_y - d)[(1/d) + k_1 \Omega(\beta) + k_2 \{(r - d)/r\}\eta(r)]\\
		& \hspace*{-5pt}  = (\bar{e}_y/d) + (\bar{e}_y - d)[k_1 \Omega(\beta) + k_2 \{(r - d)/r\} \eta(r)]\\
		\hspace*{-5pt} \dot{\bar{e}}_y &= -\bar{e}_x [(1/d) + k_1 \Omega(\beta) + k_2 \{(r - d)/r\}\eta(r)]. 
	\end{align}
\end{subequations}

\begin{theorem}\label{thm_control_design_2}	
Consider the robot-target engagement shown in Fig.~\ref{fig_problem_description} concerning the robot model \eqref{system_model_reparametrized_new} obtained under time-reparametrization \eqref{time_reparametrization}. Suppose that Assumption~\ref{assumption} and Assumption~\ref{assumption_velocity} hold. Let the model \eqref{system_model_reparametrized_new} be governed by the controller \eqref{control_2_reparametrized}. Then, $r_i < r(\tau) < r_o$ for all $\tau \geq 0$, and $\bar{e}_x(\tau) \to 0$, $\bar{e}_y(\tau) \to 0$, $\beta(\tau) \to 0$ and $r(\tau) \to d$ as $\tau \to \infty$, for all initial conditions $\mathcal{Z}_r$.
\end{theorem}

\begin{proof}
Consider the following candidate Lyapunov function
	\begin{equation}\label{V2}
		V_2(\bar{e}_x, \bar{e}_y, e_r, \beta) = V_1(\bar{e}_x, \bar{e}_y, e_r)  + dk_1\int_{0}^{\beta} \Omega(s) ds,
	\end{equation}
	where $V_1$ is given by \eqref{V1}. Note that $V_2$ is positive-definite and satisfies all the requirements indicated in Lemma~\ref{blf_lemma}, specifically, because of the property {\bf (P2)} of $\Omega$. Following the steps similar to the proof of Theorem~\ref{thm_control_design_1}, the time derivative of \eqref{V1}, along the error dynamics \eqref{error_dynamics_new}, is obtained as $\dot{V}_1 = -dk_1 \bar{e}_x\Omega(\beta)$. Therefore, the time derivative of \eqref{V2}, along \eqref{error_dynamics_new}, is given by $\dot{V}_2 = -dk_1 \bar{e}_x\Omega(\beta) + dk_1\dot{\beta}\Omega(\beta)$. Now, using \eqref{beta_dynamics}, we have $\dot{V}_2 = -dk_1 \bar{e}_x\Omega(\beta) + dk_1 (-\kappa \beta + \bar{e}_x)\Omega(\beta) = -d\kappa k_1 \beta \Omega(\beta)$. Since $\dot{V}_2 \leq 0$ and Assumption~\ref{assumption} holds, it immediately follows from Lemma~\ref{blf_lemma} that $r_i < r(\tau) < r_o$ for all $\tau \geq 0$ in $\mathcal{Z}_r$ (similar to the proof of Theorem~\ref{thm_control_design_1}). Along the similar steps, one can further conclude that there exists a positive constant $\lambda_2$ such that $V_2(\bar{e}_x, \bar{e}_y, e_r, \beta) \leq \lambda_2$ for all initial conditions in $\mathcal{Z}_r$. Note that the set $\Lambda_2 \triangleq \{[\bar{e}_x, \bar{e}_y, e_r, \beta]^\top \in \mathbb{R}^2 \times (-\delta_a, \delta_b) \times \mathbb{R} \mid V_2(\bar{e}_x, \bar{e}_y, e_r, \beta) \leq \lambda_2\}$ is compact and positively invariant, since $V_2$ is positive definite and continuously differentiable, and $\dot{V}_2 \leq 0$ in $\mathcal{Z}_r$.  Now, using LaSalle's invariance principle \cite[Corollary~4.1, pg. 128]{khalil2002control}, it can be concluded that all the trajectories of \eqref{system_model_reparametrized_new} converge to the largest invariant set $\Delta_2$ contained in the set $\Gamma_2 \triangleq \{[\bar{e}_x, \bar{e}_y, e_r, \beta]^\top \mid \dot{V}_2 = 0\} \subset \Lambda_2$, as $\tau \to \infty$. Note that $\dot{V}_2 = 0$ only if $\beta = 0$. That is, $\beta \equiv 0$ and $\dot{\beta} \equiv 0$ in $\Gamma_2$, which on substitution in \eqref{beta_dynamics}, gives $\bar{e}_x \equiv 0$ and $\dot{\bar{e}}_x \equiv 0$ in $\Gamma_2$. For the system to maintain $\dot{V}_2 = 0$ on $\Delta_2 \subset \Gamma_2$, it must hold from \eqref{error_dynamics_new} that $\dot{\bar{e}}_y \equiv 0$ and $1 + (\bar{e}_y(\tau) - d)u \equiv 0$ on $\Delta_2$, where $u$ is given by \eqref{control_2_reparametrized}. Now, the remaining analysis follows the same steps as in the proof of Theorem~\ref{thm_control_design_1}, and hence, omitted for brevity. Consequently, the largest invariant set $\Delta_2 \subset \Gamma_2 \subset \Lambda_2$ contains only the desired equilibrium where $\bar{e}_x(\tau) \to 0$, $\bar{e}_y(\tau) \to 0$, $\beta(\tau) \to 0$, and $r(\tau) \to d$, as $\tau \to \infty$. Hence, proved.
\end{proof}

\begin{corollary}\label{cor_control_2}
Consider the robot-target engagement shown in Fig.~\ref{fig_problem_description} concerning the robot model \eqref{system_model} and suppose that Assumption~\ref{assumption} and Assumption~\ref{assumption_velocity} hold. Let the robot \eqref{system_model} be governed by the controller \eqref{control_2}. Then, $r_i < r(t) < r_o$ for all $t \geq 0$, and $\bar{e}_x(t) \to 0$, $\bar{e}_y(t) \to 0$, $\beta(t) \to 0$ and $r(t) \to d$ as $t \to \infty$, for all initial conditions in $\mathcal{Z}_r$.
\end{corollary}

\begin{remark}\label{remark_initial_conditions}
Unlike \eqref{control_omega_final}, the controller \eqref{control_2} relies solely on range information and an auxiliary dynamic variable $z$, which can be initialized arbitrarily. Importantly, our approach requires only an assumption on the initial range $r(0)$, as stated in Assumption~\ref{assumption} and imposes no restriction on the initial heading angle $\theta(0)$. This stands in contrast to \cite{jain2019trajectory}, which imposes constraints on both $r(0)$ and $\theta(0)$, and whose controller implementation requires both $r(t)$ and $\theta(t)$. Further, the design functions $\Omega$ used in \eqref{control_1} and \eqref{control_2} admit a broad class of functions, including potentially unbounded ones, as described in Remark~\ref{remark_design_function}. 
\end{remark}

\section{Analysis of Post-Design Signals}
Exploiting \eqref{V1} and \eqref{V2}, this section obtains the bounds on the radial error $e_r(t)$ and the range $r(t)$ in a restricted sense and also shows the boundedness of the proposed controllers \eqref{control_omega_final} and \eqref{control_2}. In this direction, we present the following two theorems corresponding to the above cases of availability/non-availability of the range-rate information. 

\begin{theorem}[Availability of range and range-rate measurements]\label{thm_bounds_case1}
Under the conditions given in Corollary~\ref{cor_control_1}, the following statements hold for all $t \geq 0$ within the set $\mathcal{Z}_r$:
\begin{enumerate}[leftmargin=*]
	\item[(a)] The transformed error vector $[\bar{e}_x(t), \bar{e}_y(t)]^\top$ satisfies $\bar{e}_x^2(t) + \bar{e}_y^2(t) \leq 2V_1(0)$, where $V_1(0) \triangleq V_1(\bar{e}_x(0), \bar{e}_y(0), e_r(0))$, is obtained from \eqref{V1}. 
	\item[(b)] The radial error $e_r(t)$ and the range $r(t)$ remain in the compact sets: $e_r(t) \in [-\delta_a \Upsilon_1, \ \delta_b \Upsilon_1]$, and $r(t) \in [r_i + \delta_a(1 - \Upsilon_1), \ r_o - \delta_b(1 - \Upsilon_1)]$, respectively, where $\Upsilon_1 \triangleq \sqrt{1 - {\rm e}^{-(2/dk_2) V_1(0)}}$ is a constant.  
\end{enumerate}
\end{theorem}

\begin{figure*}[t!]
	\centering{\hspace*{-0.3cm}
		\subfigure[Trajectory]{\includegraphics[width=0.27\textwidth]{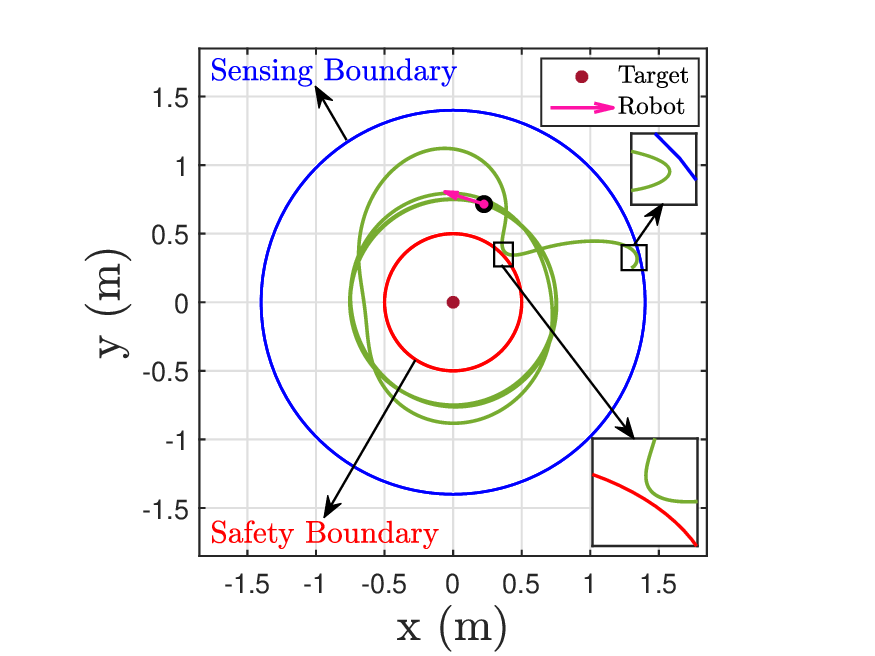}
			\label{traj_range_only}}\hspace*{-0.8cm}
		\subfigure[Radial error]{\includegraphics[width=0.27\textwidth]{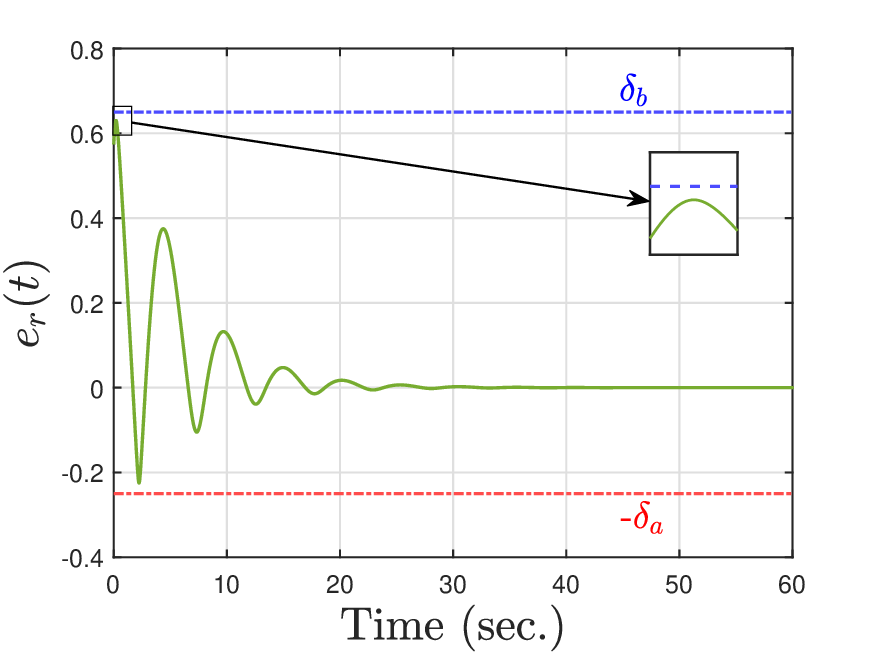}\label{error_range_only}}\hspace*{-0.5cm}
		\subfigure[Range]{\includegraphics[width=0.27\textwidth]{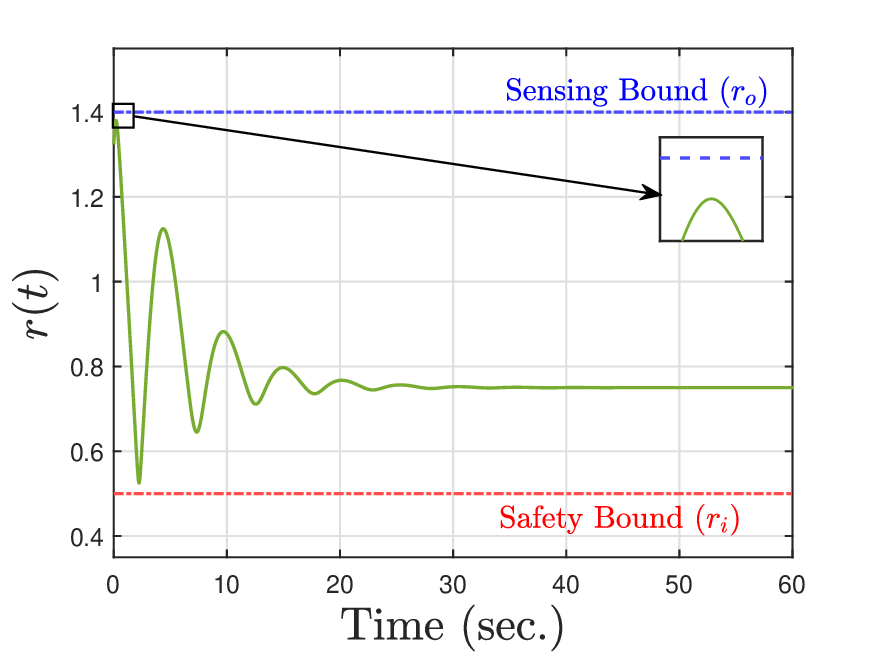}\label{distance_range_only}}\hspace*{-0.5cm}
		\subfigure[controller]{
			\includegraphics[width=0.27\textwidth]{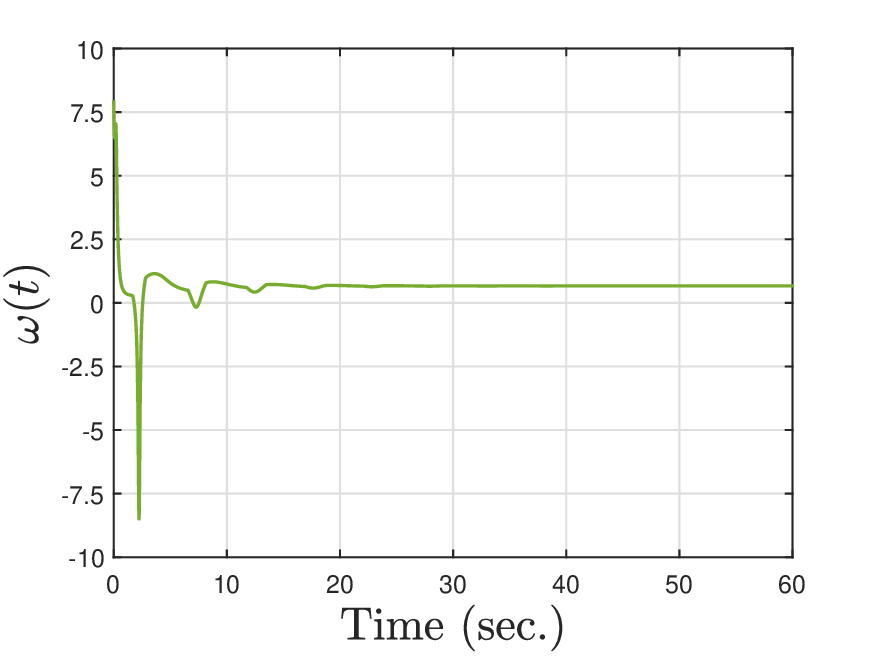}\label{control_range_only}}
		\caption{Robot's trajectory and time evolution of radial error \eqref{radial_error}, range, and controller \eqref{control_2a}, under range-only controller \eqref{control_2}.}
		\label{Range_only}}
\end{figure*} 

\begin{proof}
Since $\dot{V}_1 \leq 0$ according to Theorem~\ref{thm_control_design_1}, $V_1(\bar{e}_x(\tau), \bar{e}_y(\tau), e_r(\tau)) \leq V_1(0), \forall \tau \geq 0$ in $\mathcal{Z}_r$. Following Lemma~\ref{lem_trajectory_equivalence}, this implies that $V_1(\bar{e}_x(t), \bar{e}_y(t), e_r(t)) \leq V_1(0)$ for all $t \geq 0$ in $\mathcal{Z}_r$. Consequently, it follows from \eqref{V1} that $V_e(\bar{e}_x(t), \bar{e}_y(t)) \leq V_1(0)$ and $dk_2V_r(e_r(t)) \leq V_1(0)$ for all $t \geq 0$ in $\mathcal{Z}_r$. Now, substituting for $V_e$ and $V_r$ from \eqref{Ve} and \eqref{Vr}, respectively, the proof follows along similar steps as in \cite[Theorem~3, pg.~924]{tee2009barrier}, and is omitted for brevity. 
\end{proof}

\begin{theorem}[Availability of range-only measurements]\label{thm_bounds_case2}
Under the conditions given in Corollary~\ref{cor_control_2}, the following statements hold for all $t \geq 0$ within the set $\mathcal{Z}_r$:
	\begin{enumerate}[leftmargin=*]
		\item[(a)] The transformed error vector $[\bar{e}_x(t), \bar{e}_y(t)]^\top$ satisfies $\bar{e}_x^2(t) + \bar{e}_y^2(t) \leq 2V_2(0)$, where $V_2(0) \triangleq V_2(\bar{e}_x(0), \bar{e}_y(0), e_r(0), \beta(0))$, is obtained from \eqref{V2}. 
		\item[(b)] The radial error $e_r(t)$ and the range $r(t)$ remain in the compact sets: $e_r(t) \in [-\delta_a \Upsilon_2, \ \delta_b \Upsilon_2]$, and $r(t) \in [r_i + \delta_a(1 - \Upsilon_2), \ r_o - \delta_b(1 - \Upsilon_2)]$, respectively, where $\Upsilon_2 \triangleq \sqrt{1 - {\rm e}^{-(2/dk_2) V_2(0)}}$ is a constant.
		\item[(c)] The modulus of the signal $\beta(t)$ in \eqref{beta_dynamics} remains bounded by $|\beta(t)| \leq \alpha^{-1}_1(V_2(0)/dk_1)$, for some class $\mathcal{K}_{\infty}$ function $\alpha_1$, satisfying property {\bf (P2)} of the design functions $\Omega$.
		\item[(d)] The modulus of the auxiliary signal $z(t)$ in \eqref{control_2b} remains bounded by $\underline{\chi} \leq |z(t)| \leq \overline{\chi}$, where $\underline{\chi}  = (1/\kappa)[0.5(r_i + \delta_a(1 - \Upsilon_2))^2 - \alpha^{-1}_1(V_2(0)/dk_1)]$ and $\overline{\chi} = (1/\kappa)[0.5(r_o - \delta_b(1 - \Upsilon_2))^2 + \alpha^{-1}_1(V_2(0)/dk_1)]$.
\end{enumerate}
\end{theorem}

\begin{proof}
Since $\dot{V}_2 \leq 0$ in $\mathcal{Z}_r$, the proofs of part (a) and part (b) are the same as in Theorem~\ref{thm_bounds_case1} by replacing $V_1(0)$ by $V_2(0)$. To prove part (c), it can be inferred from \eqref{V2} that $dk_1\int_{0}^{\beta(\tau)} \Omega(s) ds \leq V_2(0)$ for all $\tau \geq 0$ in $\mathcal{Z}_r$. Using property {\bf (P2)}, it can be concluded that $\alpha_1(|\beta(\tau)|) \leq \int_{0}^{\beta(\tau)} \Omega(s) ds \leq {V_2(0)}/{dk_1}$ for some class $\mathcal{K}_{\infty}$ function $\alpha_1$ for all $\tau \geq 0$ in $\mathcal{Z}_r$. Alternatively,  $\alpha_1(|\beta(\tau)|) \leq {V_2(0)}/{dk_1} \implies |\beta(\tau)| \leq \alpha^{-1}_1({V_2(0)}/{dk_1})$ for all $\tau \geq 0$ in $\mathcal{Z}_r$. Now, the result immediately follows applying Lemma~\ref{lem_trajectory_equivalence}. Further, to prove part (d), replacing $\beta(\tau) = -\kappa z(\tau) + \gamma(\tau)$ in the preceding inequality, implies $|-\kappa z(\tau) + \gamma(\tau)| \leq \alpha^{-1}_1({V_2(0)}/{dk_1})$. Since $\big{|} |\kappa z(\tau)| - |\gamma(\tau)| \big{|} \leq |-\kappa z(\tau) + \gamma(\tau)|$, it follows from the former inequality that $\big{|} |\kappa z(\tau)| - |\gamma(\tau)| \big{|} \leq  \alpha^{-1}_1({V_2(0)}/{dk_1})$, which implies that $(1/\kappa)[|\gamma(\tau)| - \alpha^{-1}_1(V_2(0)/dk_1)] \leq |z(\tau)| \leq (1/\kappa)[|\gamma(\tau)| + \alpha^{-1}_1(V_2(0)/dk_1)]$ for all $\tau \geq 0$ in $\mathcal{Z}_r$. Since $\gamma = (1/2)r^2$, it follows from the bound on $r(\tau)$ in part (b) that $\frac{1}{2}(r_i + \delta_a(1 - \Upsilon_2))^2 \leq |\gamma(\tau)| \leq \frac{1}{2}(r_o - \delta_b(1 - \Upsilon_2))^2$ for all $\tau \geq 0$ in $\mathcal{Z}_r$. Consequently, the result follows after applying the preceding bounds on $|\gamma(\tau)|$ and Lemma~\ref{lem_trajectory_equivalence}.  
\end{proof}

\begin{remark}[Boundedness of proposed controllers]\label{control_bounds}
Note that the controllers \eqref{control_omega_final} and \eqref{control_2} remain bounded along any solution trajectory within $\mathcal{Z}_r$, because of the following: (i) the term $\eta(r)$, given by \eqref{eta}, is finite since $-\delta_a < e_r(t) < \delta_b$ for all $t \geq 0$ in $\mathcal{Z}_r$, as established in Corollary~\ref{cor_control_1} and Corollary~\ref{cor_control_2}, (ii) the design function $\Omega$ remains finite due to its bounded argument for all $t \geq 0$ in $\mathcal{Z}_r$ in both the cases: (a) For controller \eqref{control_omega_final}, the boundedness of $\Omega(r\dot{r}/v)$ follows from the fact that $\bar{e}_x$ is bounded in a compact set from Theorem~\ref{thm_bounds_case1}, and so is $r\dot{r}$, using \eqref{range_dynamics}. (b) For controller \eqref{control_2}, the boundedness of $\Omega(-\kappa z + \gamma) = \Omega(\beta)$ follows from Theorem~\ref{thm_bounds_case2} (part (c)). 
\end{remark}

\begin{figure}[t!]
	\centering{
		\subfigure[Trajectory (Case I)]{\includegraphics[width=0.27\textwidth]{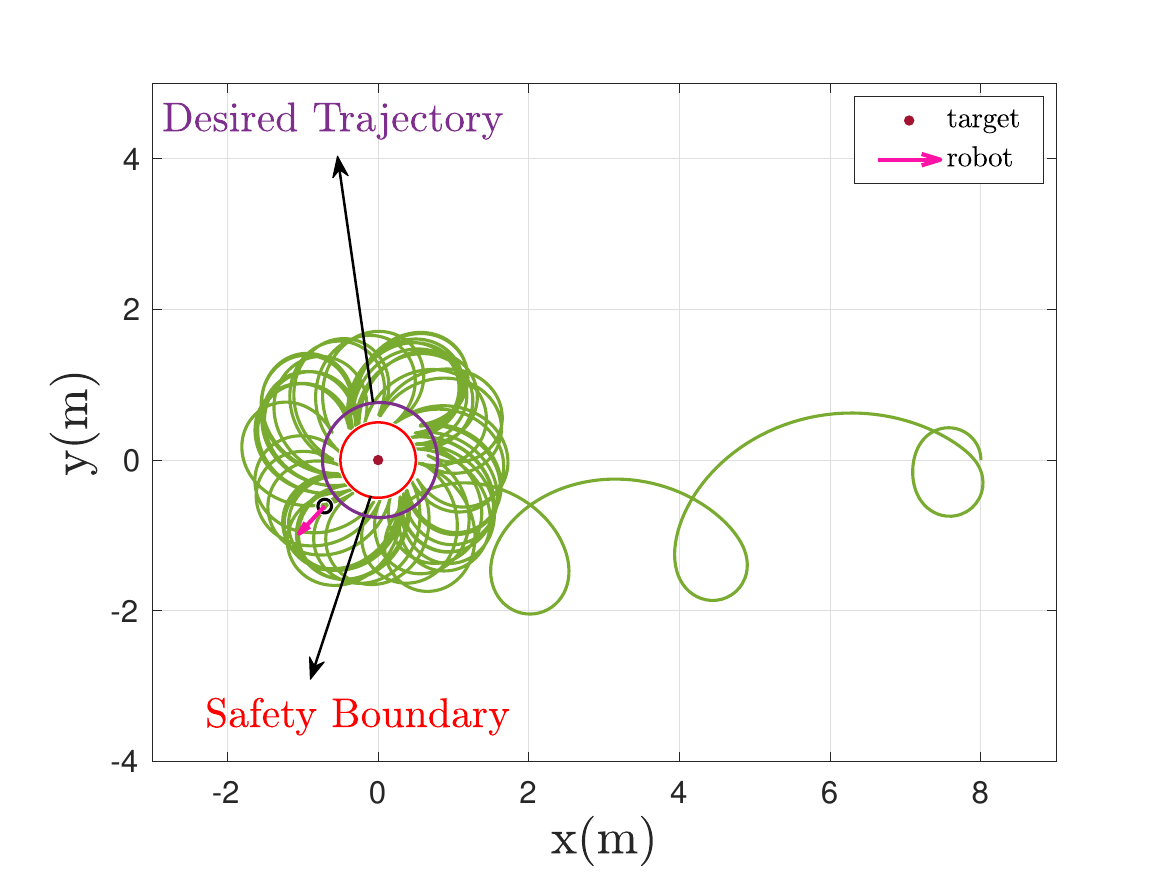}\label{traj_case_1_collision_avoidance}}\hspace*{-15pt}
		\subfigure[Trajectory (Case II)]{\includegraphics[width=0.27\textwidth]{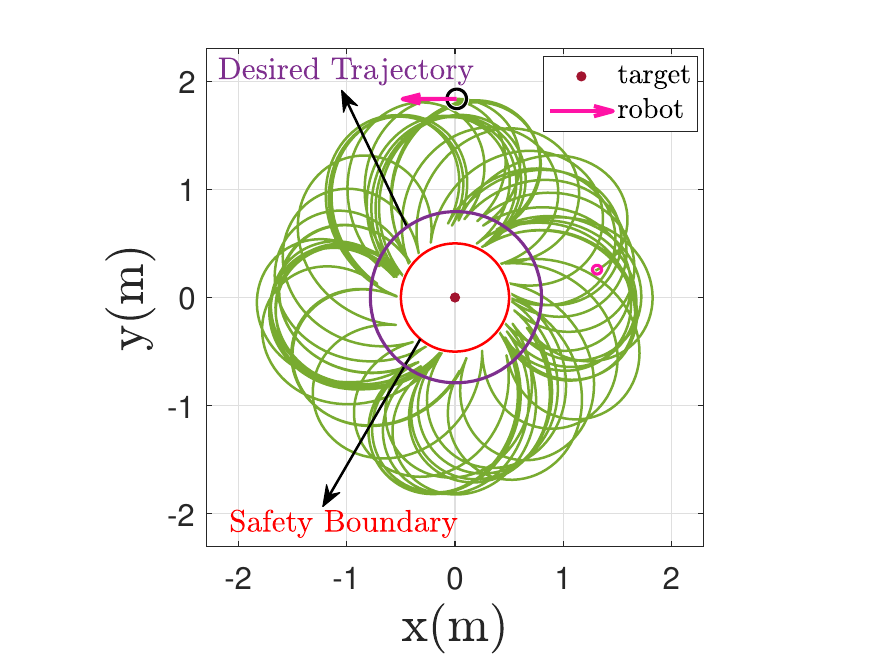}\label{traj_case_2_collision_avoidance}}
		\caption{Robot's trajectories using approach from \cite{wang2024target} considering safety boundary as an obstacle.}
		\label{traj_collision_avoidance}
		}
\end{figure}

\section{Simulation and Experimental Results}
In principle, our approach is applicable even if the robot begins arbitrarily close to the two boundaries, provided Assumption~\ref{assumption} holds. However, in practice, a real robotic system must be initialized with a sufficient safety margin to account for its physical structure. We illustrate both these aspects through simulations and experiments as follows: 

\subsection{Simulations}
Consider a robot moving at a constant speed of $v = 0.5$ (m/s) with an initial position and heading angle of $[x(0), y(0)]^\top = [1.3, 0.25]^\top$ (m) and $\theta(0) = 25^\circ$, respectively. Let the robot be deployed to circumnavigate the target $T$, located at $[x_T, y_T]^\top = [0, 0]^\top$, at a desired radius $d = 0.75$ (m), while satisfying safety and sensing constraints with $r_i = 0.5$ (m) and $r_o = 1.4$ (m), respectively. Clearly, Assumption~\ref{assumption} is satisfied, since $r_i = 0.5 < r(0) = 1.3238 < r_o = 1.4$. Using \eqref{barriers}, the barriers are obtained as $\delta_a = d - r_i = 0.25$ and $\delta_b = r_o-d = 0.65$. Assuming a linear (unbounded) design function $\Omega(\mu) = \mu$ for simplicity, we simulated the range-only controller \eqref{control_2} with controller gains $k_1 = 14, \ k_2 = 0.5, \ \kappa = 7$, initializing the auxiliary variable as $z(0) = 0$. The results are plotted in Fig.~\ref{Range_only}, where Fig.~\ref{traj_range_only} shows that the robot circumnavigates the target at the desired distance and successfully satisfies both the safety and sensing constraints. Fig.~\ref{error_range_only} and Fig.~\ref{distance_range_only} depict the time evolution of the radial error $e_r(t)$ and the range $r(t)$, respectively. These satisfy their respective bounds as obtained in Theorem~\ref{thm_control_design_2} and eventually approach their desired values. Fig.~\ref{control_range_only} plots the time evolution of the controller \eqref{control_2a} which approaches the steady state value of $v/d \approx 0.67$ (rad/s), as desired. As stated in Remark~\ref{remark_design_function}, employing the saturation function $\Omega = \mu/\sqrt{1 + \mu^2}$ yields qualitatively similar results (plots omitted for brevity). Additionally, on simulating under the controller \eqref{control_1}, we observed relatively fewer oscillations and faster convergence as compared to controller \eqref{control_2}, because of the availability of both range and range-rate information. The plots are again omitted for brevity.  

\begin{figure}[t!]
	\centering{
		\subfigure[Khepera IV ground robot]{\includegraphics[width=0.22\textwidth]{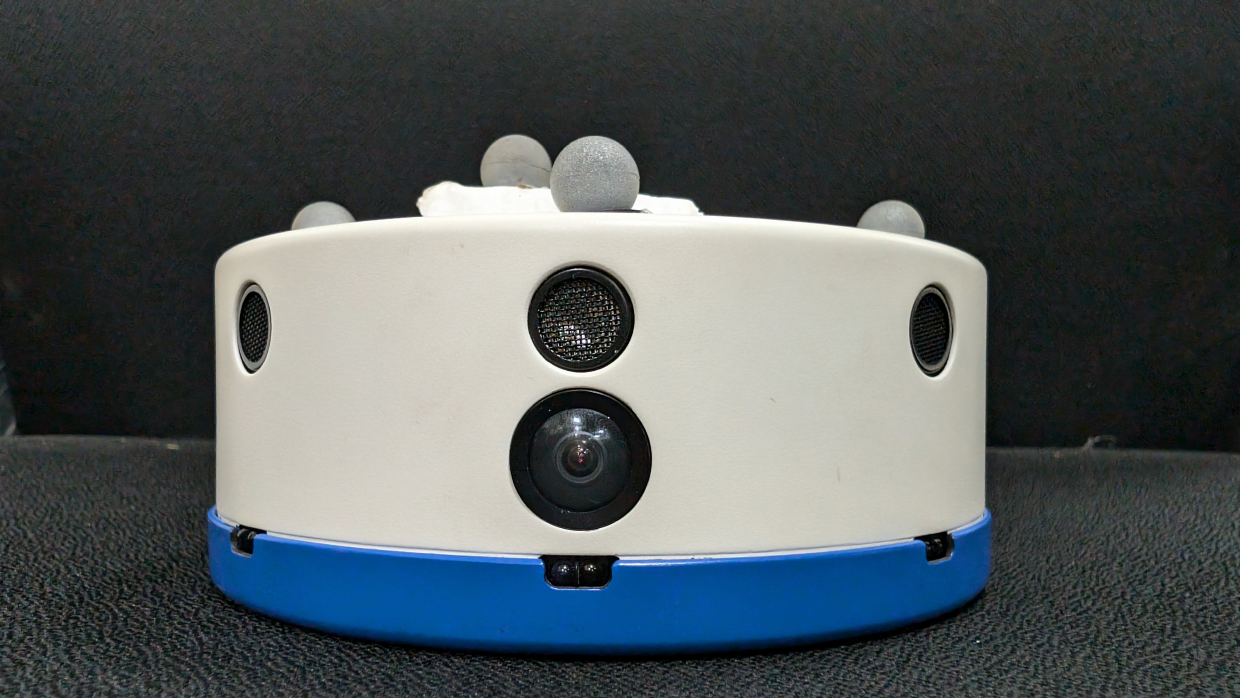}\label{khepera}} \hspace*{0.1cm}
		\subfigure[Motion capture cameras]{\includegraphics[width=0.22\textwidth]{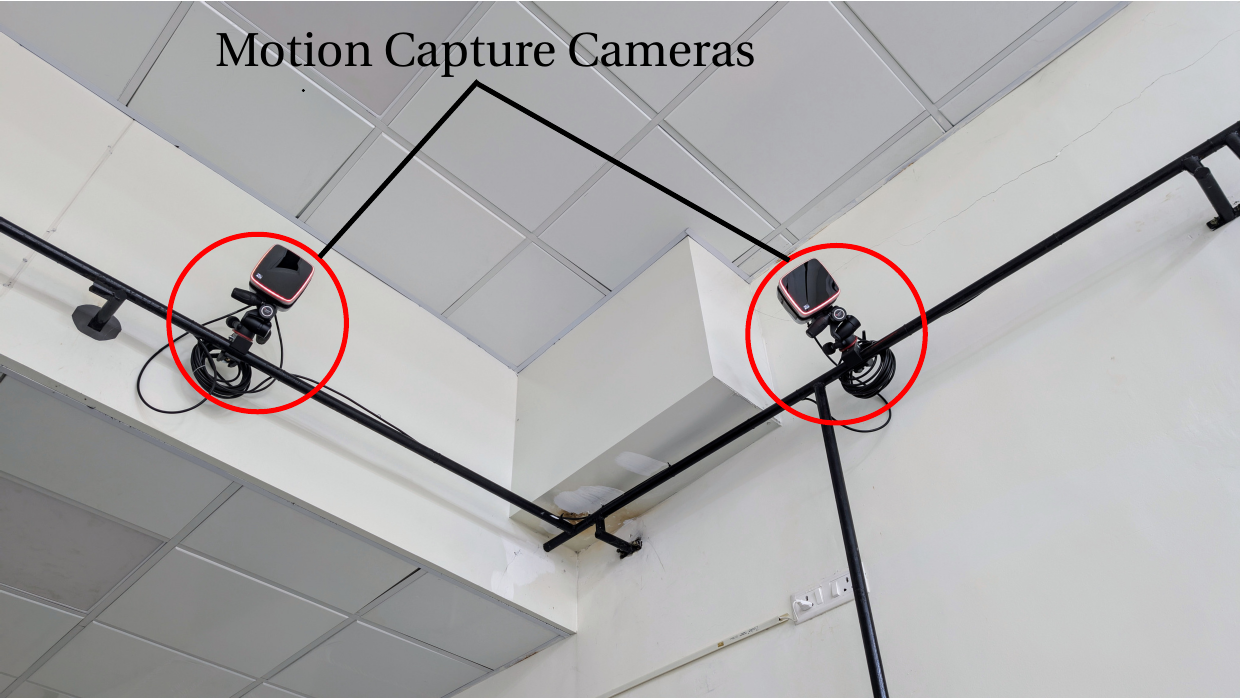}\label{cameras}}
		\caption{Khepera IV robot and MoCap validation setup.}
		\label{lab}}
\end{figure}

\begin{figure*}[t!]
	\centering{
		\subfigure[Trajectory at $t=26$ (s)]{\includegraphics[width = 0.185\textwidth]{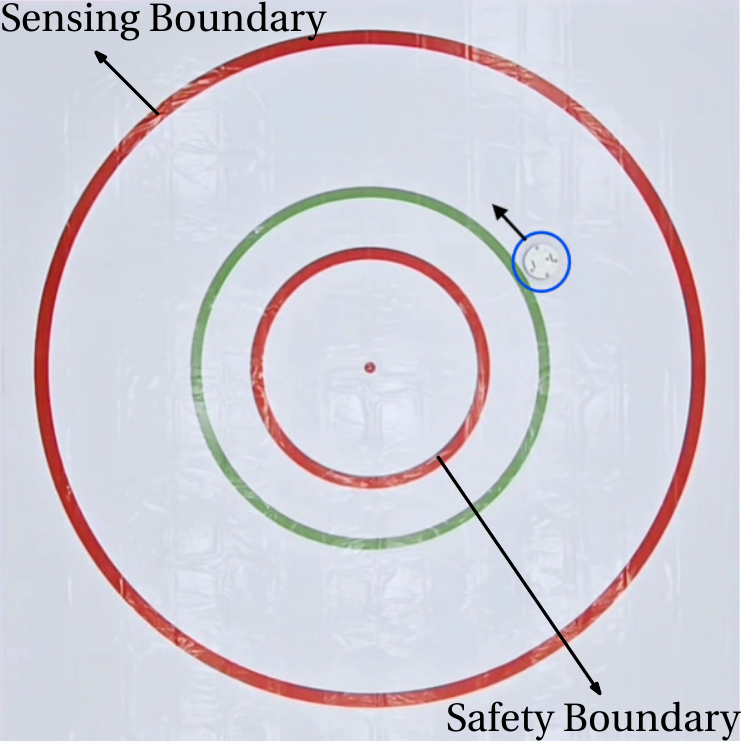}
			\label{topview_worr}} \hspace*{0.35cm}
		\subfigure[From task manager]{\includegraphics[width = 0.192\textwidth]{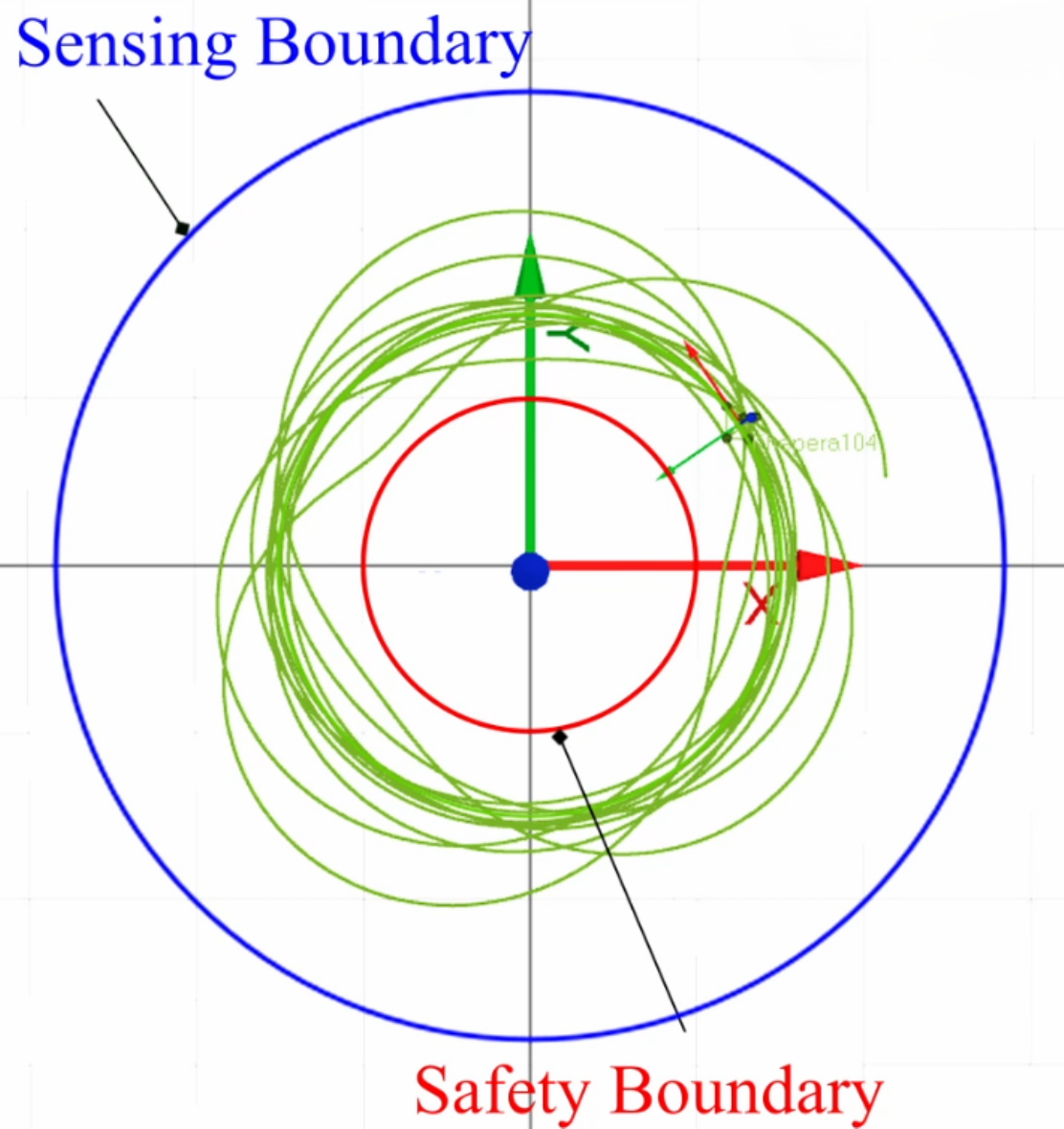}
			\label{qtm_experi_worr}}\hspace*{0.15cm}
		\subfigure[Range]{\includegraphics[width=0.27\textwidth]{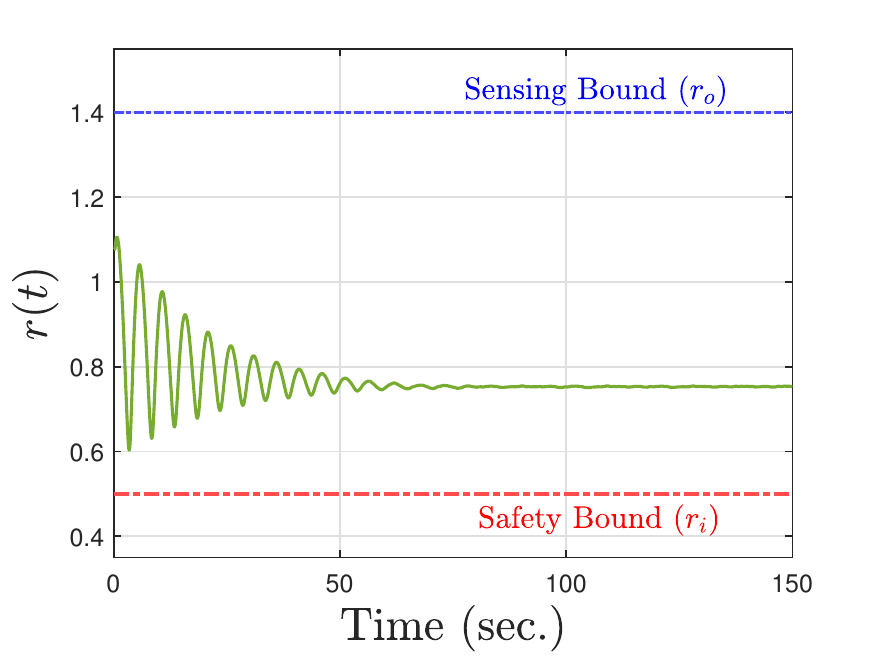}
			\label{experi_radius_worr}}\hspace*{-0.1cm}
		\subfigure[controller]{\includegraphics[width=0.27\textwidth]{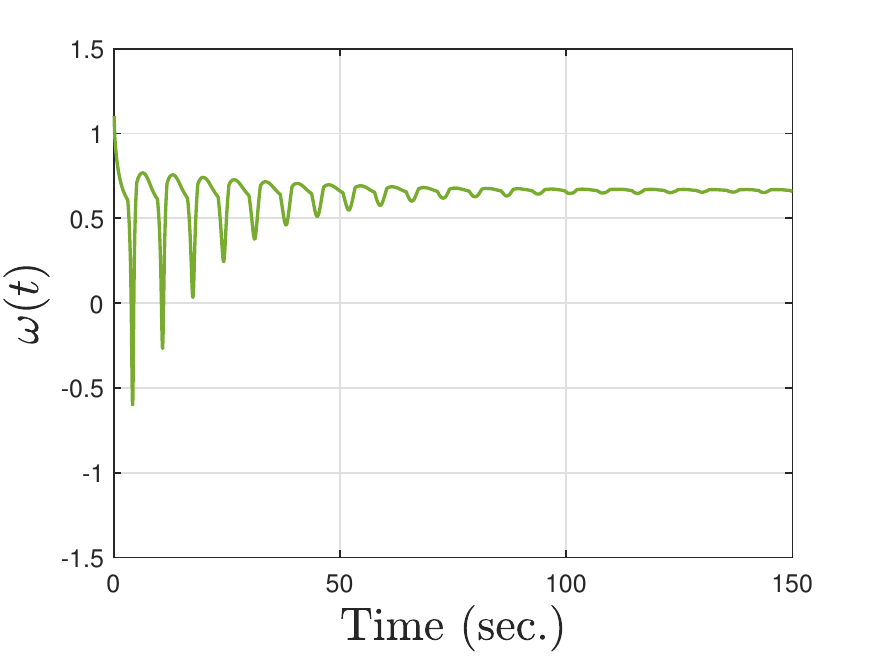}
			\label{experi_control_worr}}
		\caption{Experimental results for Khepera IV robot using range only measurements.}	
		\label{experi_range_only}}
\end{figure*}

\begin{remark}
	Note that \cite{wang2024target} also addresses collision avoidance for circular obstacles, which are essentially convex with respect to the robot's motion. One may wonder whether the approach of \cite{wang2024target} can be extended to our problem setting by treating the inner and outer circular boundaries as obstacles. Since the outer boundary is non-convex with respect to the robot's motion in our setup, we provide a comparison that considers the inner boundary as an obstacle. It was observed that the approach of \cite{wang2024target} exhibits sensitivity to the proximity between the desired circumnavigation path and the inner boundary. Specifically, when the radial separation between the desired trajectory and the inner boundary becomes small, the controller in \cite{wang2024target} may fail to stabilize the robot on the intended circular path. This behavior is illustrated in Fig.~\ref{traj_collision_avoidance}, which shows robot trajectories under two different initial conditions using the same inner and desired boundaries as in Fig.~\ref{traj_range_only}.
\end{remark} 

\subsection{Experiments}
We performed experiments using a \emph{Khepera IV} differential-drive ground robot (see Fig.~\ref{khepera}) with the help of a motion capture (MoCap) validation setup comprising overhead cameras (see Fig.~\ref{cameras}). The MoCap system is operated via a task manager where the data is recorded and processed for feedback. We considered that the Khepera IV begins with the initial position and heading angle of $[x(0), y(0)]^\top = [1.05, 0.20]^\top$ and $\theta(0) = 90^\circ$, respectively (satisfying Assumption~\ref{assumption} for the same boundary constraints as in simulation subsection). Since Khepera is a differential drive robot, we obtained the speeds $v_\ell$ and $v_r$ of its left and right vehicles from the linear and angular speeds of the (unicycle) robot \eqref{system_model} as follows: $v_r = v + (d_w/2)\omega, \ v_\ell = v - (d_w/2)\omega$ where $d_w = 10.54$ (cm) is the distance between the two wheels of Khepera robot. The robot was operated at a constant linear speed $v = 0.5$ (m/s), following the hardware limit of maximum wheel speed of $0.814$ (m/s). Upon implementing the controller \eqref{control_2} with controller gains $k_1 = 3, \ k_2 = 1, \ \kappa = 4$, we captured the experimental results as shown in Fig.~\ref{experi_range_only}. Fig.~\ref{topview_worr} provides a top-down view of the robot's motion, highlighting its position at $t = 26$ seconds. The robot's path, recorded throughout the experiment from the task manager, is illustrated in Fig.~\ref{qtm_experi_worr}. Further, Fig.~\ref{experi_radius_worr} and Fig.~\ref{experi_control_worr} display the range and the performance of the controller \eqref{control_2a}, respectively. The video of the conducted experiments can be found on \url{https://youtu.be/XoNbWddjNvY}.   


\section{Conclusion and Future Remarks}\label{Conclusion}
Using the concepts of the range-based dynamic output feedback controller and ABLF, this paper investigated the problem of safe and secure circumnavigation of a hostile target by a unicycle robot. The proposed controller relied exclusively on range measurements and ensured that the robot's trajectories remained bounded within a predefined annular region around the target by appropriately selecting the barriers $\delta_a$ and $\delta_b$. The (local) asymptotic stability of the closed-loop system was rigorously established, and analytical bounds were obtained for various intermediate signals, including the proposed controllers. The core of the controller design lied in effectively combining the two potential functions \eqref{Ve} and \eqref{Vr}, where $V_e$ was derived from the transformed relative positional error components in the Cartesian plane, while $V_r$ was constructed from the positional error relative to the desired circumnavigation radius in polar coordinates. It is worth noticing that both these potentials minimize at the desired equilibrium where $r \to d$, however, none of these individually can guarantee both the objectives simultaneously, as: (i) \eqref{Ve} does not account for the desired safety and sensing constraints, (ii) \eqref{Vr} does fulfill both the stabilizing and constraints requirements, it alone cannot be used for controller design, since its derivative depends on $\dot{r}$ which is not (directly) controllable via $\omega$.  

As discussed in Remark~\ref{control_bounds}, the control bounds in the proposed design depend on the initial conditions, unlike \cite{wang2024target}, due to the requirement that the initial range satisfy Assumption~\ref{assumption}. While pre-specified control bounds may potentially be achieved by incorporating actuator saturation under additional restrictions on both $\theta(0)$ and $r(0)$ \cite{jain2019trajectory}, guaranteeing such bounds using only range measurements remains an interesting direction for future work.


\appendix

\begin{proof}[Proof of Lemma~\ref{lem_trajectory_equivalence}]
Since $d \tau / dt = v(t) > 0$, the map $t \mapsto \tau(t)$ is strictly increasing and hence invertible. Also, $\tau(t) = \int_{0}^{t} v(s) ds \geq \int_{0}^{t} v_{\min} ds = v_{\min}t$, implying that $\tau(t) \to \infty$ as $t \to \infty$. Let $t = t(\tau)$ denote the inverse map. For any solution $\zeta(t)$ of the original $t$-system, define the corresponding trajectory in the reparametrized time scale by $\zeta(\tau) = \zeta(t(\tau))$ (i.e., using the inverse map). Then, using the chain rule $\frac{d\zeta}{d\tau} = \frac{d\zeta}{dt} \frac{dt}{d\tau} = v(t)f(\zeta)\frac{1}{v(t)} = f(\zeta)$. Conversely, let $\zeta(\tau)$ be a solution of $\tau$-system $d\zeta/d\tau = f(\zeta)$, then $\zeta(t) = \zeta(\tau(t))$ satisfies $\frac{d\zeta}{dt} = \frac{d\zeta}{d \tau} \frac{d \tau}{dt} = f(\zeta)v(t) = v(t)f(\zeta)$. Hence $\zeta(t)$ is a solution of the original $t$-system. Therefore, the two systems possess a one-to-one correspondence between trajectories and differ only in time parametrization. Finally, if $\zeta(\tau) \to \zeta^{\star}$ as $\tau \to \infty$, then, because $\tau(t) \to \infty$, $\zeta(t) = \zeta(\tau(t)) \to \zeta^{\star}$ as $t \to \infty$. Hence, proved. 	
\end{proof}


\bibliographystyle{IEEEtran}
\bibliography{References}

\end{document}